\documentclass[preprint]{imsart}

\usepackage[utf8]{inputenc}
\usepackage{amsthm,amsmath,bbm}
\usepackage{hyperref}

\usepackage{xcolor,listings,multicol} 
\usepackage{amssymb} 
\usepackage[]{enumitem}
\usepackage{graphicx}

\usepackage[figuresleft]{rotating}

\startlocaldefs

\theoremstyle{plain}
\newtheorem{theorem}{Theorem}[section]
\newtheorem{lemma}[theorem]{Lemma}

\theoremstyle{definition}

\newtheorem{remark}{Remark}

\newcommand{\bpi}{{\boldsymbol{\pi}}}
\newcommand{\bPi}{{\boldsymbol{\Pi}}}

\newcommand{\bX}{{\boldsymbol{X}}}

\newcommand{\bProb}{{\boldsymbol{{q}}}}

\endlocaldefs

\begin{document}

\begin{frontmatter}

\title{Likelihood test in permutations with bias\\
\large \emph{Premier League} and \emph{La Liga}:
surprises during the last 25 seasons}
\runtitle{Permutations with bias}

\author{\fnms{Giacomo} \snm{Aletti}\corref{}\ead[label=e1]{giacomo.aletti@unimi.it}\thanksref{t1}}
\address{ADAMSS Center, Universit\`a degli Studi di Milano, \\
20131 Milano, Italy
\\ \printead{e1}}
\thankstext{t1}{Member of ``Gruppo Nazionale per il Calcolo Scientifico (GNCS)'' 
of the Italian Institute ``Istituto Nazionale di Alta Matematica (INdAM)''.
This work was partially developed during a visiting
research period at the School of Mathematical Science,
Fudan University, Shanghai, China. The author thanks for the hospitality.} 

\runauthor{Giacomo Aletti}

\begin{abstract}
In this paper, we introduce the models of \emph{permutations with bias},
which are random permutations of a set, biased by some preference values. 
We present a new parametric test, together with an efficient way to calculate its
$p$-value. The final tables of the English and Spanish major soccer leagues 
are tested according to this new procedure, to discover whether these results 
were aligned with expectations.
\end{abstract}

\begin{keyword}[class=MSC]
\kwd[Primary ]{62D05}
\kwd[; secondary ]{65C10; 68A55}
\end{keyword}

\begin{keyword}
\kwd{Permutations with bias}
\kwd{Expected preferences}
\kwd{Likelihood ratio test}
\kwd{soccer unpredictability}
\end{keyword}

\end{frontmatter}

\section{Introduction}
Ranks are everywhere in our lives. We rank, we are ranked, and we sometimes depend on ranks.
We rank our interests (conscientiously or not), we rank friends according our preferences, 
we rank colleagues, etc\ldots\
At the same time, we are ranked at high school, during the university carrier (and after). 
In addition, we hope that our team will be well-ranked
at the end of the season, and we continuously reorder our priorities, based, also, on these ranks.
A rank is essentially a permutation of a group of thinks, and usually it is not totally predictable.

In statistics, permutation procedures are becoming more and more popular for 
constructing sampling distributions,
by reordering the observed data. 
Basically, random shuffles of the data are used to get the correct distribution of a 
suitable test statistic under a given null hypothesis. 
It is usually much more computationally intensive than standard statistical tests.
Non-parametric tests are 
often proposed for testing the homogeneity of two or more populations (see, recently,
\cite{Ush17}). For functional data, the importance of the permutation approach
is discussed in, e.g., \cite{Cor2014}.
In \cite{Pes12}, the method of nonparametric combination 
of dependent permutation tests is reviewed together its main properties.
A specific permutation procedure is also used in 
variable selection (see, e.g., \cite{Sab15}).

The idea at the base of the classical permutation procedure is that all the permutations 
are equally likely to be expressed, at least in principle. In other words,
exchangeable-like assumptions are assumed in the sample, under the null hypothesis.
Conversely, in this paper, we work with random permutations of a set
which are assumed to be biased by some ``preference values''.
Consequently, the rank of each objects of the set is expected
to be higher if its preference value is higher, see \cite{Stern90}.
This random procedure models, for example, the final rank of a league, which
is biased by the strength of each team at the beginning of the season.

A similar idea may be found in the context of discrete choice model,
where under study is the process that leads 
to an agent's choice among a set of possible actions,
see \cite{Kenneth09} for a recent book. The agent's preferences
may by inferred by a researcher, by estimating its utility function. 
The final choice is based on these preferences:
the higher the preference is, more likely the corresponding action will be chosen.
This paper extends that idea, 
by considering not only the ``final choice'' of the agent, but all
the rank of the agent, as in the models given, for example, in \cite{Stern90}. 

This framework of discrete choice model
has recently inspired a new technique for random variable
generation, see \cite{A1art}. Here, we use also the idea at the base of this
technique for an exact efficient simulation of the whole process of 
permutation with bias. Based on this result, a likelihood ratio test
may be efficiently defined to test whether the hypothesized preferences
were exact or not, or, in other words, If the final rank 
neglects the expectations. We then apply this new theory to the data 
of two of the most known European soccer league: the Spanish
\emph{La Liga} and the English \emph{Premier League}, by comparing
the final tables of the last $25$ years with the expectations at the beginning of
each year.

The paper is structured as follows. Section~\ref{sec:perm_bias} introduces the methodological novelties
of this paper. At the beginning, we define the model of \emph{permutations with bias},
then we introduce the (parametric) likelihood ratio test together with the definition
of the exact $p$-value of the test, and we conclude the section with the description
of the efficient Monte Carlo procedure to evaluate the $p$-value. Section~\ref{sec:PL_conf}
deals with the application of the methodology. It starts by describing the models 
that describes the link between the team ranking at the beginning of the season, 
and the expected performance of that team at the end of the season. These expected performances
are used as biases in our model, and in the second part of the section we perform
the hypothesis test for each season and for each league, and we present the results.
In Section~\ref{sec:concl} we give the conclusions of the paper, while in the appendix we derive the 
correctness of the Monte Carlo procedure.

\section{Permutations with bias}\label{sec:perm_bias}

In the sequel, the permutations of the set $\{1,\ldots,n\}$ are denoted with bold Greek letters, so that
$\bpi = (\pi_1, \ldots ,\pi_n)$ is such that $\pi_i \in \{1,\ldots,n\}$ for any $i$ and
$\pi_i \neq \pi_j$ if $i\neq j$. We use uppercase bold Greek letters, as $\bPi$, to denote
random objects with value in the set of permutations. 
The vectors as $\bProb =(q_1, \ldots, q_n)$ will be always defined
with strictly positive elements, if not differently stated. Accordingly, it is possible to evaluate the natural
logarithm, denoted here by $\log(\cdot)$, to each of its elements.

At time $t=0$,
we assume to have $n$ different objects, labelled with their natural index $\{1,\ldots,n\}$.
In addition, the sequence $\bProb=({q}_1, \ldots, {q}_n)$ of positive preference values 
is associated to our objects.

We work with \emph{permutations with bias}, that are particular ordered samplings without replacement
of our objects, where the selection probability depends on the preference values. 
The result is a random permutation $\bPi$ with law \eqref{eq:randPerm}, obtained with
the follow procedure.

At each time $t\in \{1, \ldots,n\}$, an object is selected between the existing ones
with probability \emph{proportional} to its preference value ${q}_i$, independently on the past. 
Its label $\pi_t$ is assigned to the $t$-th rank,
and the object is discharged.
At the end, the random permutation $\bpi = (\pi_1, \ldots, \pi_n)$
of the first $n$ numbers is obtained with probability $P_\bProb$ or likelihood $L$ given by
\begin{equation}\label{eq:randPerm}
P_\bProb(\bPi = \bpi ) = \prod_{i=1}^n \frac{ {q}_{\pi_i} }{ \sum_{j=i}^{n} {q}_{\pi_j} } =:
L(\bProb | \bpi).
\end{equation}

\begin{remark}\label{rem:logDependence}
We underline that \eqref{eq:randPerm} is not sensitive to multiplicative factors.
In fact, if ${r}_i = c{q}_i$, then
\[
\prod_{i=1}^n \frac{ {r}_{\pi_i} }{ \sum_{j=i}^{n} {r}_{\pi_j} }
=
\prod_{i=1}^n \frac{ c{q}_{\pi_i} }{ \sum_{j=i}^{n} c{q}_{\pi_j} }
= \prod_{i=1}^n \frac{ {q}_{\pi_i} }{ \sum_{j=i}^{n} {q}_{\pi_j} }.
\]
\end{remark}

\subsection{A likelihood ratio test}

In a permutation with bias, it is possible to define the following likelihood ratio test
\begin{equation}\label{eq:LRtest}
\begin{aligned}
H_0 &: \bProb=\bProb_0,\\
H_1 &: \bProb \neq \bProb_0 ,
\end{aligned}
\end{equation}
where the likelihood ratio test statistic is
\[
\Lambda(\bpi)=
\frac{L(\bProb_0 | \bpi) }{\sup\{ L(\bProb | \bpi ) : \bProb \in\mathbb{R}^n_+\,\}}.
\]
The likelihood ratio is small if the alternative model is better than the null model 
and the likelihood ratio test provides the decision rule as follows:
\begin{align*}
&\text{Do not reject }H_0 & \text{if }\Lambda > c^* ;
\\
&\text{Reject }H_0 & \text{if }\Lambda < c^* ;
\\
&\text{Reject }H_0 \text{ with probability }q& \text{if }\Lambda = c^* .
\end{align*}
By symmetry arguments, it is obvious that $\sup\{ L(\bProb | \bpi ) : \bProb \in\mathbb{R}^n_+\,\}$
is a constant function of $\bpi$, and hence the critical region may be computed with $L(\bProb_0 | \bpi) $
(instead of $\Lambda$) with a constant $c$ (instead of $c^*$). 
The values $c, c^*, q$ are usually chosen to have the desired significance, in that
\[
q P_{\bProb_0}(\Lambda=c^* ) + P_{\bProb_0}(\Lambda < c^* ) = 
q P_{\bProb_0}(L=c ) + P_{\bProb_0}(L < c ) = \text{Significance level of the test}. 
\]
When ${q}_1 = \cdots = {q}_n $, then all the objects are equally likely to be extracted, and,
as expected, we get 
\(P_{{q}_1 = \cdots = {q}_n }(\bPi = \bpi ) = \frac{1}{n!}\), for any $\bpi$. As a consequence, in this uniform case, 
we obtain a test which is independent on the observed $\bpi$.
This is not the case when the terms ${q}_i$ are different, that we will our case study.

Given the ordered sequence \({q}_{\sigma_1} \leq {q}_{\sigma_2} \leq \cdots \leq  {q}_{\sigma_n}\) 
of $\bProb$, 
even if the problem $ \{\bpi \colon L \leq c\} $ is, in general, intractable, it is obvious that
the sequence $({\sigma_1}, \sigma_2, \ldots, \sigma_n)$ belongs to the 
critical region and the sequence $({\sigma_n}, \sigma_{n-1}, \ldots, \sigma_1)$ to the acceptance one.
In fact, for any permutation $\bpi$ and $i = 1,\ldots, n$, we have
\(
\sum_{j=i}^{n} {q}_{\sigma_j} \geq 
\sum_{j=i}^{n} {q}_{\pi_j} \geq
\sum_{j=i}^{n} {q}_{n+1-\sigma_j} 
\) and hence, by \eqref{eq:randPerm},
\[
L(\bProb | ({\sigma_1}, \sigma_2, \ldots, \sigma_n)) = \min_{\bpi^*} L(\bProb | {\bpi^*})
\leq L(\bProb | \bpi) \leq 
\max_{\bpi^*} L(\bProb | {\bpi^*}) =
L(\bProb | ({\sigma_n}, \sigma_{n-1}, \ldots, \sigma_1)) .
\]
However, in principle, once a certain $\bpi^*$ is observed, it is possible to define the $p$-value 
in the classical way
\begin{equation}\label{eq:LRpValue}
p\text{-value} = \sum_{ \bpi \colon L(\bProb | \bpi) < L(\bProb | \bpi^*)  } L(\bProb | \bpi)  
= \sum_{ \bpi \colon L(\bProb | \bpi) < L(\bProb | \bpi^*)  } P_{\bProb}(\bPi = \bpi )  ,
\end{equation}
where $\bProb = \bProb_0$ for the test given in \eqref{eq:LRtest}.


\subsection{Efficient simulation}

To compute \eqref{eq:LRpValue} for a given value of $\bProb$ and an observed sequence $\bpi^*$, 
a Monte Carlo procedure is used here to calculate an approximated empirical $p$-value in the following way.

In the spirit of \cite{Stern90} and, more recently, 
\cite{A1art}, it is possible to generate a random permutation $\bpi$ in the following way.
A random vector $\bX =(X_1, \ldots, X_n)$ with independent components is generated, where 
each $X_i $ is distributed as an exponential random variable with parameter ${q}_i$.
The random permutation is defined as the indexes of the order statistics:
$(X_{\pi_1}, \ldots, X_{\pi_n})=(X_{(1)}, \ldots, X_{(n)})$. 
In the Appendix, we show that this generation has the same law of \eqref{eq:randPerm}
(as also given in \cite[Equation (4)]{Stern90} and in the reference therein):
\begin{equation}\label{eq:geomMin}
P_{\bProb}( X_{\pi_1} < \cdots < X_{\pi_n} ) = P_{\bProb}(\bPi = \bpi ) .
\end{equation}
This result extends also that of \cite{A1art}, where it is shown that $P_{\bProb}(\Pi_1=\pi)= {{q}_{\pi} }/\sum_1^n {q}_i $.
Note that, since we are interested only in the order of the indexes, 
we may simulate $Y_{i} = \log( X_{i} )$, 
and we compare directly $\{Y_{i}, i=1,\ldots,n\}$ (see, again, \cite[Section 4]{Stern90}).
%
To do so, we start with a table of independent uniform
random variables $\{U_{i,m}, i= 1,\ldots, n, m=1,\ldots,M\}$.
Then we compute \(Y^{(m)}_{i} = \log( -\log( U_i) ) -\log({q}_{i}) \), and we register the ordered indexes in 
$\bpi^{(m)}=\{\pi_{i}^{(m)}, i= 1,\ldots, n, m=1,\ldots,M\}$, so that
\[
Y_{\pi_{1}^{(m)}}^{(m)} < Y_{\pi_{2}^{(m)}}^{(m)} < \cdots < Y_{\pi_{n}^{(m)}}^{(m)} , \qquad \text{for any }m=1,\ldots,M.
\]
The $\log$-likelihood $\ell_m = \ell_{\bpi^{(m)}}$ of each simulated sequence is hence registered without the common additive factor
$\sum_i \log({q}_i)$, in the following way\footnote{Many softwares have the built-in function \texttt{cumsum}.
It is more convenient to store $\bpi^{(m)}$ in reverse order, i.e.\ 
$\boldsymbol{p} = ({q}_{\pi_{n}^{(m)}}, {q}_{\pi_{n-1}^{(m)}}, \ldots, {q}_{\pi_{1}^{(m)}} )$, 
and to compute $\ell_m = \texttt{-sum(log(cumsum(} \boldsymbol{p}\texttt{)))}$.}
\[
\ell_{m} = - \sum_{i=1}^n \log\Big(\sum_{j=i}^n {q}_{\pi_{j}^{(m)}} \Big), \qquad \text{for any }m=1,\ldots,M.
\]
The comparison of these latter with
$\ell_{\bpi^*} = - \sum_{i=1}^n \log (\sum_{j=i}^n {q}_{\pi^*_{j}} )$, that is 
computed for the observed sequence $\bpi^*$, gives
\begin{equation}\label{eq:compPvalue}
\hat{q}= \frac{\#\{m \colon \ell_m < \ell_{\bpi^*} \}}{M}.
\end{equation} 

Summing up, we generate i.i.d.\ random permutations $\bpi^{(m)}$ 
with common distribution $\bPi$ (by \eqref{eq:geomMin}) and then we compute
the sequence of \(\{\ell_m, m=1,\ldots,M\}\), which are themselves realizations of i.i.d.\ random variables.
Note that \eqref{eq:compPvalue} gives the empirical $p$-value, since
\[
P_\bProb( \bpi\in\bPi\colon \ell_\bpi < \ell_{\bpi^*} ) = P _\bProb( \bpi\in\bPi\colon L(\bProb | \bpi) < L(\bProb | \bpi^*)  ) =  
\sum_{ \bpi \colon L(\bProb | \bpi) < L(\bProb | \bpi^*)  } P_\bProb(\bPi = \bpi) .
\]

\section{\emph{Premier League} and \emph{La Liga}: season results and expectations}\label{sec:PL_conf}

In this section, we analyse two soccer national leagues from 1992-93 (first \emph{Premier League} season)
to 2016-17, to test whether the final tables 
were expectable or not. For each season, we 
compute the \emph{a priori} expectations of the probability of winning each season for each
team, 
and we compare it with \emph{the final} obtained ranks of the teams. 

\subsection{Elo ratings and expected probability of winning the season}
The \emph{World Football Elo Rating} (ER) is becoming more and more popular due to its significant power of prevision,
see, e.g., \cite{Las13}. ER is based on the Elo rating system 
and includes modifications to take various soccer-specific variables into account.

The difference in the ERs between two teams serves as a predictor of the outcome of a match  with a logistic model. 
In other words, the logarithm of the winning probability of each match is essentially proportional to ER, 
up to a factor for the advantage of the home team.
Obviously, there is more uncertainty in the result of a single game than in 
the averaged result of a season, and hence we must recalibrate the ER based model. 
We show in a moment that the expected winning probability of the season of each team, in logarithm scale, is again proportional to 
its ER.

\begin{figure}
\includegraphics[width=.75\textwidth]{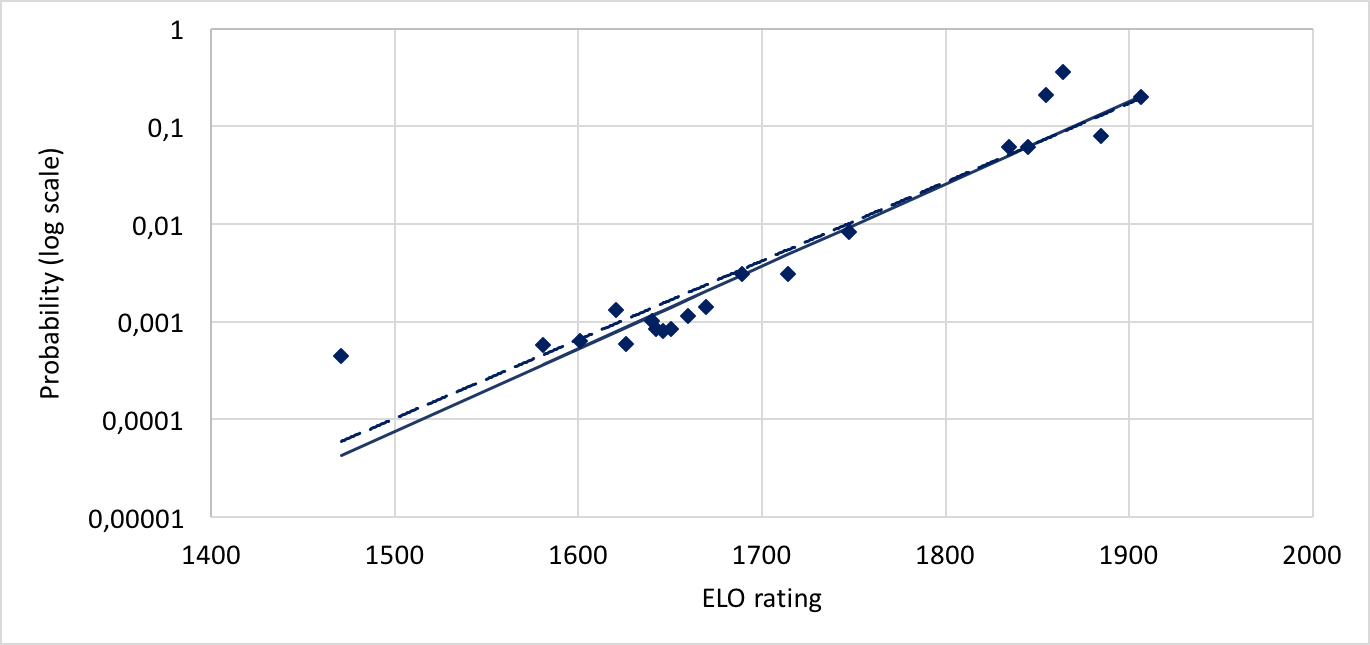}
\caption{Comparison between ELO ratings and logarithm of expected winning probabilities 
based on the odds of the principal online betting players. The ordinary least square linear fitting 
(dashed line) is plotted together with a robust one (solid line).}\label{fig:ELO-lnPB}
\end{figure}

To achieve this task, we have downloaded the odds for the winner of the Premiere League of all the big competitors in the UK online betting system at a day of summer, a quiet period. 
We have computed the averaged expected probability of winning of each team, and we have compared with the correspondent ELO rating. In Figure~\ref{fig:ELO-lnPB}, the scatter plot shows a good linear model
(Multiple R-squared:  0.9026,	Adjusted R-squared:  0.8972,  p-value $< 10^{-9}$). We have calibrated the model with a robust regression fitting (using an M estimator, see \cite{Hub09}) to reduce the contribution
of the evident outlier. Note that the slope parameter is the
sole interesting one, as underlined also in Remark~\ref{rem:logDependence}.

The expectation of the winning probability for each team is hence computed considering its ER at the 1st of October of the corresponding season. In this way, we think to have included the ELO adjustments due to the summer markets, which are reflected in the initial part of the season. Summing up, we are assuming that ERs of 1st of October are good predictions of the 
initial expectations of the people for the teams of that season. The relative expected probabilities of winning each season
are shown in Table~\ref{TabP:EN2}-\ref{TabP:EN1} and in Table~\ref{TabP:ES2}-\ref{TabP:ES1} for the \emph{Premier League} and 
\emph{La Liga}, respectively, together with the ranks obtained by the teams at the end of the season.

\subsection{Unbelievable seasons}
To evaluate the unexpected results of the two national leagues, we have modelled each final season ranks as a
permutation with bias. The likelihood test \eqref{eq:LRtest} is performed, with $\bProb_0$ being the 
relative expected probabilities of winning.
A significant $p$-value (less than $0.05$), computed as in \eqref{eq:compPvalue}, 
reveals that either the expectations were wrong or that the result is highly
surprising. In both cases, from a personal perception, the lower is the $p$-value, the higher is felt strange the final table.
In information theory (see \cite{Pie12}), the fact that an event is \emph{informative} is measured thorough its 
self-information or surprisal, and computed as the opposite of the natural logarithm of the probability of the event. 
The scale is given in the natural unit of information (nat).

In Figure~\ref{fig:Surpr} it is plotted the time series of the surprisals of the $p$-values. As known, 
the result of \emph{Leicester} has made the 2014-15 season exceptional, the third more unpredictable English season
in the Premier League era. It should be stressed that not only the winner, 
but all the teams contribute to the unpredictability of the final table according to their initial strength and final ranks. 
This is the case for the Spanish 2003-04 season, where the debates of both
the \emph{Celta de Vigo} and the \emph{Real Sociedad} ($19$th and $15$th in the final rank, respectively) 
made this season ``unbelievable''.

\begin{figure}
\includegraphics[width=.75\textwidth]{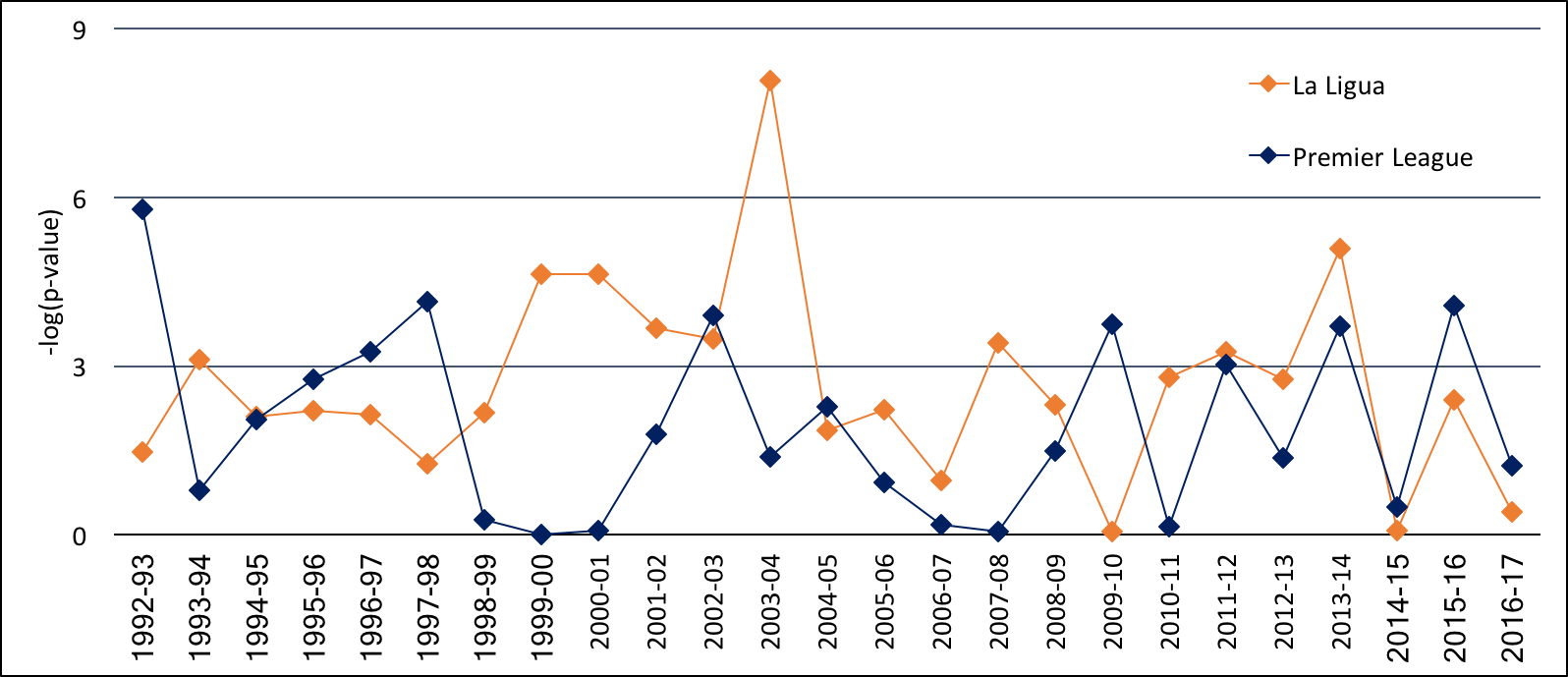}
\caption{Surprisal of the seasons' final tables, plotted as self-information of the $p$-values.
A surprisal of more than 3 correspond to a p-value less than $0.05$.}\label{fig:Surpr}
\end{figure}

\section{Conclusions}\label{sec:concl}
In this paper, we have presented a new test for permutations with bias, 
that are ordered samplings without replacement 
where the selection probability depends on a preference value of each unit.
Since the sample size is given by $n!$ possible permutations and analytic expressions are not given, 
we have provided a method to compute Monte Carlo $p$-values in an efficient way.

As an example, we have studied the results of the Spanish \emph{La Liga} and of the English 
\emph{Premier League}, since the foundation of the latter.
By analysing Elo ranks of the teams at the beginning of each season, 
we could find the rational expectations for the different seasons.
We have tested whether the final tables were in accordance with the expectations, and we found that 
more than $30\%$ of the seasons had unpredictable results, in both the Spanish and English league.
That's soccer!


\appendix
\section{Mathematical derivation of \eqref{eq:geomMin}}
In this Appendix, we give the mathematical proof of the accuracy of our Monte Carlo procedure.
We begin with a lemma.
\begin{lemma}\label{lem:proofEq}
For $\bpi = (\pi_1, \ldots, \pi_n)$,
let $g_\bpi:(0,1)\times \{1,\ldots,n\}\to\mathbb{R}_+$ be the function so defined: 
\[
g_\bpi(u,k) = 
\begin{cases}
\int_0^u {q}_{\pi_k} v^{{q}_{\pi_k}-1} g_\bpi(v,k+1) dv & \text{if }k <n.
\\
\int_0^u {q}_{\pi_n} v^{{q}_{\pi_n}-1} dv & \text{if }k = n.
\end{cases}
\]
Then 
\[
g_\bpi(u,k) = \prod_{i=k}^n \frac{{q}_{\pi_i}u^{{q}_{\pi_i}}}{ \sum_{j=i}^n {q}_{\pi_j} },
\]
and, in particular, $g_\bpi(1,1) = \prod_{i=1}^n \frac{{q}_{\pi_i}}{ \sum_{j=i}^n {q}_{\pi_j} }$.
\end{lemma}
\begin{proof}
For $k=n$, it is a standard computation. For $k<n$, by backward induction,
\begin{align*}
g_\bpi(u,k) & = \int_0^u {q}_{\pi_k} v^{{q}_{\pi_k}-1} 
\prod_{i=k+1}^n \frac{{q}_{\pi_i} v^{{q}_{\pi_i}}}{ \sum_{j=i}^n {q}_{\pi_j} } dv
\\
& = \frac{\prod_{i=k}^n {q}_{\pi_i}}{ \prod_{i=k+1}^n \sum_{j=i}^n {q}_{\pi_j} } \int_0^u v^{{q}_{\pi_k}-1+\sum_{i=k+1}^n {q}_{\pi_i} } dv
\\
& = \frac{\prod_{i=k}^n {q}_{\pi_i}}{ \prod_{i=k+1}^n \sum_{j=i}^n {q}_{\pi_j} } 
\frac{[v^{\sum_{i=k}^n {q}_{\pi_i}}]^u_0 }{ \sum_{i=k}^n {q}_{\pi_i} } 
\\
& = \frac{\prod_{i=k}^n {q}_{\pi_i} u^{\sum_{i=k}^n {q}_{\pi_i}}}{ \prod_{i=k}^n \sum_{j=i}^n {q}_{\pi_j} } 
= \prod_{i=k}^n \frac{{q}_{\pi_i} u^{{q}_{\pi_i}}}{ \sum_{j=i}^n {q}_{\pi_j} }. &\qedhere
\end{align*} 
\end{proof}
The desired result is a consequence of the previous lemma, as shown below.
\begin{proof}[Proof of \eqref{eq:geomMin}]
We recall that, given a geometric random variable $X$ with parameter $q$, the random variable
$U = \exp(-q X)$ is uniformly distributed on $(0,1)$. 
Accordingly, if we transform the random vector $\bX$, we obtain
so that
\begin{equation*}
P( X_{\pi_1} < \cdots < X_{\pi_n} )
= P\Big( -\frac{\log(U_{\pi_1})}{{q}_{\pi_1}} < \cdots < -\frac{\log(U_{\pi_n})}{{q}_{\pi_n}} \Big)
= P\Big( U_{\pi_1}^{\frac{1}{{q}_{\pi_1}}} > \cdots > U_{\pi_n}^{\frac{1}{{q}_{\pi_n}}} \Big),
\end{equation*}
where $(U_{\pi_1}, \ldots, U_{\pi_n})$ is a vector of i.i.d.\ random variables uniformly distributed on $(0,1)$.
As a consequence, the random vector $(U_{\pi_1}^{\frac{1}{{q}_{\pi_1}}} , \cdots , U_{\pi_n}^{\frac{1}{{q}_{\pi_n}}} )$
has density
\[
f( u_1, \ldots, u_n ) = \prod_{i=1}^n {q}_{\pi_i} u_i ^{{q}_{\pi_i}-1} \mathbbm{1}_{(0,1)}(u_i),
\]
and hence, by Lemma~\ref{lem:proofEq},
\begin{multline*}
P( X_{\pi_1} < \cdots < X_{\pi_n} )
 = \int_0^1 {q}_{\pi_1} u_1^{{q}_{\pi_1}-1} \Big( \int_0^{u_1} {q}_{\pi_2} u_2^{{q}_{\pi_2}-1} 
\Big( \int_0^{u_2} \cdots  du_3 \Big) du_2 \Big) du_1 \\
= g_\bpi(1,1) =P_{\bProb}(\bPi= \bpi) .\qedhere \\
\end{multline*}
\end{proof}


\begin{sidewaystable}
    \centering
    \caption{\emph{Premier League}: relative probabilities of the expectation of winning and final ranks (seasons from 2004-05
    to 2016-17). At bottom: $p$-value of the test \eqref{eq:LRtest}.}\label{TabP:EN2}
\begin{tabular}{|c*{13}{|@{\;}r@{\;}}|}
\hline
Name & 2016-17 & 2015-16 & 2014-15 & 2013-14 & 2012-13 & 2011-12 & 2010-11 & 2009-10 & 2008-09 & 2007-08 & 2006-07 & 2005-06 & 2004-05 \\
\hline
AFC Bournemouth & $  1.91^{[9]}$ & $  1.66^{[16]}$ &&&&&&&&&&&\\
Arsenal & $206.01^{[5]}$ & $136.80^{[2]}$ & $ 303.70^{[3]}$ & $1174.44^{[4]}$ & $224.95^{[4]}$ & $ 109.15^{[3]}$ & $ 290.47^{[4]}$ & $ 979.44^{[3]}$ & $ 622.93^{[4]}$ & $ 781.49^{[3]}$ & $ 591.20^{[4]}$ & $ 378.16^{[4]}$ & $2014.09^{[2]}$ \\
Aston Villa && $  1.00^{[20]}$ & $   4.00^{[17]}$ & $  27.74^{[15]}$ & $  5.20^{[15]}$ & $  17.23^{[16]}$ & $  17.97^{[9]}$ & $  39.59^{[6]}$ & $  54.64^{[6]}$ & $  36.04^{[6]}$ & $  12.86^{[11]}$ & $   4.04^{[16]}$ & $  30.81^{[10]}$ \\
Birmingham City &&&&&&& $   3.44^{[18]}$ & $   2.38^{[9]}$ && $   3.20^{[19]}$ && $   3.25^{[18]}$ & $   9.58^{[12]}$ \\
Blackburn Rovers &&&&&& $   4.79^{[19]}$ & $  10.62^{[15]}$ & $  10.32^{[10]}$ & $  20.49^{[15]}$ & $  46.12^{[7]}$ & $  28.74^{[10]}$ & $   3.98^{[6]}$ & $   9.53^{[15]}$ \\
Blackpool &&&&&&& $   1.01^{[19]}$ &&&&&&\\
Bolton Wanderers &&&&&& $   3.02^{[18]}$ & $   3.40^{[14]}$ & $   6.88^{[14]}$ & $   5.59^{[13]}$ & $  11.71^{[16]}$ & $  32.30^{[7]}$ & $  18.87^{[8]}$ & $  27.57^{[6]}$ \\
Burnley & $  2.66^{[16]}$ && $   1.81^{[19]}$ &&&&& $   2.04^{[18]}$ &&&&&\\
Cardiff City &&&& $   5.07^{[20]}$ &&&&&&&&&\\
Charlton Athletic &&&&&&&&&&& $   3.12^{[19]}$ & $   7.67^{[13]}$ & $  14.23^{[11]}$ \\
Chelsea & $ 47.28^{[1]}$ & $250.54^{[10]}$ & $1447.59^{[1]}$ & $1246.65^{[3]}$ & $443.52^{[3]}$ & $ 436.86^{[6]}$ & $1565.41^{[2]}$ & $3252.79^{[1]}$ & $3013.48^{[3]}$ & $1048.44^{[2]}$ & $1694.72^{[2]}$ & $1609.49^{[1]}$ & $ 821.25^{[1]}$ \\
Crystal Palace & $  3.48^{[14]}$ & $  8.56^{[15]}$ & $   4.00^{[10]}$ & $   1.00^{[11]}$ &&&&&&&&& $   1.00^{[18]}$ \\
Derby County &&&&&&&&&& $   1.00^{[20]}$ &&&\\
Everton & $ 11.66^{[7]}$ & $ 29.57^{[11]}$ & $ 126.12^{[11]}$ & $ 293.77^{[5]}$ & $105.91^{[6]}$ & $  34.13^{[7]}$ & $  27.32^{[7]}$ & $ 127.50^{[8]}$ & $  44.61^{[5]}$ & $  57.15^{[5]}$ & $  16.61^{[6]}$ & $   2.19^{[11]}$ & $  15.78^{[4]}$ \\
Fulham &&&& $  13.36^{[19]}$ & $ 23.94^{[12]}$ & $  17.83^{[9]}$ & $  17.15^{[8]}$ & $  22.30^{[12]}$ & $   3.89^{[7]}$ & $   6.03^{[17]}$ & $  10.08^{[16]}$ & $   3.89^{[12]}$ & $  17.82^{[13]}$ \\
Hull City & $  1.31^{[18]}$ && $   1.13^{[18]}$ & $   1.98^{[16]}$ &&&& $   1.00^{[19]}$ & $   1.77^{[17]}$ &&&&\\
Leicester City & $ 57.86^{[12]}$ & $  5.52^{[1]}$ & $   6.13^{[14]}$ &&&&&&&&&&\\
Liverpool & $ 91.24^{[4]}$ & $ 29.67^{[8]}$ & $ 372.46^{[6]}$ & $ 278.66^{[2]}$ & $ 34.75^{[7]}$ & $  61.43^{[8]}$ & $ 109.64^{[6]}$ & $1722.73^{[7]}$ & $1230.04^{[2]}$ & $ 613.46^{[4]}$ & $ 203.99^{[3]}$ & $  80.17^{[3]}$ & $ 123.58^{[5]}$ \\
Manchester City & $321.27^{[3]}$ & $443.72^{[4]}$ & $1403.61^{[2]}$ & $1134.84^{[1]}$ & $991.19^{[2]}$ & $ 160.04^{[1]}$ & $  74.74^{[3]}$ & $  51.24^{[5]}$ & $  13.32^{[10]}$ & $  16.15^{[9]}$ & $   6.81^{[14]}$ & $  11.84^{[15]}$ & $  10.70^{[8]}$ \\
Manchester United & $ 48.21^{[6]}$ & $168.75^{[5]}$ & $ 154.68^{[4]}$ & $1567.63^{[7]}$ & $894.93^{[1]}$ & $1899.65^{[2]}$ & $1286.50^{[1]}$ & $5887.56^{[2]}$ & $2642.63^{[1]}$ & $1686.90^{[1]}$ & $ 517.62^{[1]}$ & $ 179.29^{[2]}$ & $ 560.64^{[3]}$ \\
Middlesbrough & $  1.00^{[19]}$ &&&&&&&& $   7.22^{[19]}$ & $  10.52^{[13]}$ & $  10.50^{[12]}$ & $   7.05^{[14]}$ & $  19.93^{[7]}$ \\
Newcastle United && $  1.31^{[18]}$ & $   3.16^{[15]}$ & $  18.40^{[10]}$ & $ 42.86^{[16]}$ & $  16.55^{[5]}$ & $   7.93^{[12]}$ && $   5.81^{[18]}$ & $  19.82^{[12]}$ & $  24.51^{[13]}$ & $   9.90^{[7]}$ & $  88.11^{[14]}$ \\
Norwich City && $  3.28^{[19]}$ && $  14.68^{[18]}$ & $  2.47^{[11]}$ & $   1.00^{[12]}$ &&&&&&& $   3.63^{[19]}$ \\
Portsmouth &&&&&&&& $   3.13^{[20]}$ & $  24.20^{[14]}$ & $  30.54^{[8]}$ & $   9.06^{[9]}$ & $   1.54^{[17]}$ & $   7.33^{[16]}$ \\
Queens Park Rangers &&& $   1.00^{[20]}$ && $  1.79^{[20]}$ & $   1.21^{[17]}$ &&&&&&&\\
Reading &&&&& $  3.07^{[19]}$ &&&&& $  11.70^{[18]}$ & $   8.92^{[8]}$ &&\\
Sheffield United &&&&&&&&&&& $   1.00^{[18]}$ &&\\
Southampton & $ 35.94^{[8]}$ & $ 12.37^{[6]}$ & $  31.59^{[7]}$ & $  16.27^{[8]}$ & $  1.00^{[14]}$ &&&&&&&& $   8.17^{[20]}$ \\
Stoke City & $  2.62^{[13]}$ & $  9.76^{[9]}$ & $  11.42^{[9]}$ & $  14.78^{[9]}$ & $  9.93^{[13]}$ & $  13.53^{[14]}$ & $   5.02^{[13]}$ & $   6.01^{[11]}$ & $   1.00^{[12]}$ &&&&\\
Sunderland & $  1.81^{[20]}$ & $  1.12^{[17]}$ & $   4.63^{[16]}$ & $   5.99^{[14]}$ & $ 11.78^{[17]}$ & $   4.76^{[13]}$ & $   4.32^{[10]}$ & $   4.40^{[13]}$ & $   3.42^{[16]}$ & $   3.62^{[15]}$ && $   1.29^{[20]}$ &\\
Swansea City & $  4.42^{[15]}$ & $  8.76^{[12]}$ & $  10.66^{[8]}$ & $  46.29^{[12]}$ & $  5.77^{[9]}$ & $   1.46^{[11]}$ &&&&&&&\\
Tottenham Hotspur & $ 87.42^{[2]}$ & $ 41.32^{[3]}$ & $  61.86^{[5]}$ & $ 624.44^{[6]}$ & $108.38^{[5]}$ & $  84.82^{[4]}$ & $  51.02^{[5]}$ & $  62.78^{[4]}$ & $  12.59^{[8]}$ & $  50.70^{[11]}$ & $  22.64^{[5]}$ & $  11.80^{[5]}$ & $  12.46^{[9]}$ \\
Watford & $  2.59^{[17]}$ & $  1.76^{[13]}$ &&&&&&&&& $   1.13^{[20]}$ &&\\
West Bromwich Albion & $  3.12^{[10]}$ & $  4.41^{[14]}$ & $   4.64^{[13]}$ & $  26.31^{[17]}$ & $ 17.00^{[8]}$ & $   3.75^{[10]}$ & $   4.40^{[11]}$ && $   1.90^{[20]}$ &&& $   1.00^{[19]}$ & $   1.34^{[17]}$ \\
West Ham United & $  3.73^{[11]}$ & $  4.96^{[7]}$ & $   7.49^{[12]}$ & $  14.08^{[13]}$ & $  6.30^{[10]}$ && $   2.47^{[20]}$ & $  14.60^{[17]}$ & $  10.47^{[9]}$ & $   8.40^{[10]}$ & $   2.69^{[15]}$ & $   1.48^{[9]}$ &\\
Wigan Athletic &&&&& $  6.08^{[18]}$ & $   2.90^{[15]}$ & $   1.00^{[16]}$ & $  10.23^{[16]}$ & $   6.97^{[11]}$ & $   3.62^{[14]}$ & $   3.92^{[17]}$ & $   1.53^{[10]}$ &\\
Wolverhampton Wanderers &&&&&& $   2.32^{[20]}$ & $   1.55^{[17]}$ & $   1.61^{[15]}$ &&&&&\\

\hline
$p$-value &
$0.29344$ & $0.01701$ & $0.61163$ & $0.02428$ & $0.25546$ & $0.04817$ & $0.85973$ & $0.02364$ & $0.2249$ & $0.93661$ & $0.837$ & $0.39563$ & $0.10245$ \\
\hline
\end{tabular}
\end{sidewaystable}

\begin{sidewaystable}
    \centering
    \caption{\emph{Premier League}: relative probabilities of the expectation of winning and final ranks (seasons from the foundation
    to 2003-04). At bottom: $p$-value of the test \eqref{eq:LRtest}.}\label{TabP:EN1}
\begin{tabular}{|c*{12}{|@{\;}r@{\;}}|}
\hline
Name & 2003-04 & 2002-03 & 2001-02 & 2000-01 & 1999-00 & 1998-99 & 1997-98 & 1996-97 & 1995-96 & 1994-95 & 1993-94 & 1992-93 \\
\hline
AFC Bournemouth &&&&&&&&&&&&\\
Arsenal & $ 315.86^{[1]}$ & $870.38^{[2]}$ & $ 84.03^{[1]}$ & $317.63^{[2]}$ & $ 505.03^{[2]}$ & $ 58.22^{[2]}$ & $205.46^{[1]}$ & $ 39.95^{[3]}$ & $ 39.25^{[5]}$ & $ 80.24^{[12]}$ & $ 67.78^{[4]}$ & $76.23^{[10]}$ \\
Aston Villa & $   5.94^{[6]}$ & $  5.10^{[16]}$ & $ 15.98^{[8]}$ & $ 21.14^{[8]}$ & $  17.87^{[6]}$ & $ 23.62^{[6]}$ & $ 48.46^{[7]}$ & $ 14.86^{[5]}$ & $ 22.98^{[4]}$ & $  9.90^{[18]}$ & $108.29^{[10]}$ & $13.17^{[2]}$ \\
Barnsley &&&&&&& $  1.00^{[19]}$ &&&&&\\
Birmingham City & $   9.13^{[10]}$ & $  1.00^{[13]}$ &&&&&&&&&&\\
Blackburn Rovers & $  13.52^{[15]}$ & $  7.70^{[6]}$ & $  2.03^{[10]}$ &&& $  2.27^{[19]}$ & $ 55.13^{[6]}$ & $ 14.27^{[13]}$ & $ 47.50^{[7]}$ & $ 79.77^{[1]}$ & $ 32.61^{[2]}$ & $ 1.71^{[4]}$ \\
Bolton Wanderers & $   4.19^{[8]}$ & $  1.53^{[17]}$ & $  2.12^{[16]}$ &&&& $ 15.15^{[18]}$ && $  1.00^{[20]}$ &&&\\
Bradford City &&&& $  1.06^{[20]}$ & $   1.60^{[17]}$ &&&&&&&\\
Charlton Athletic & $   4.74^{[7]}$ & $  3.20^{[12]}$ & $  3.99^{[14]}$ & $  2.35^{[9]}$ && $  1.85^{[18]}$ &&&&&&\\
Chelsea & $ 124.57^{[2]}$ & $ 53.38^{[4]}$ & $ 52.29^{[6]}$ & $ 28.66^{[6]}$ & $ 184.24^{[5]}$ & $ 22.05^{[3]}$ & $ 91.65^{[4]}$ & $ 10.14^{[6]}$ & $ 15.90^{[11]}$ & $ 16.48^{[11]}$ & $ 20.55^{[14]}$ & $ 5.22^{[11]}$ \\
Coventry City &&&& $  1.91^{[19]}$ & $   6.38^{[14]}$ & $  2.78^{[15]}$ & $ 13.17^{[11]}$ & $  1.72^{[17]}$ & $  5.04^{[16]}$ & $  5.71^{[16]}$ & $ 12.87^{[11]}$ & $ 4.56^{[15]}$ \\
Crystal Palace &&&&&&& $  5.34^{[20]}$ &&& $  2.84^{[19]}$ && $ 5.08^{[20]}$ \\
Derby County &&& $  1.28^{[19]}$ & $  2.00^{[17]}$ & $   4.70^{[16]}$ & $  4.05^{[8]}$ & $ 13.79^{[9]}$ & $  1.60^{[12]}$ &&&&\\
Everton & $   8.93^{[17]}$ & $  3.12^{[7]}$ & $  3.63^{[15]}$ & $  4.96^{[16]}$ & $  13.91^{[13]}$ & $  1.52^{[14]}$ & $  9.92^{[17]}$ & $ 17.98^{[15]}$ & $  9.36^{[6]}$ & $  1.65^{[15]}$ & $ 18.43^{[17]}$ & $ 6.40^{[13]}$ \\
Fulham & $  10.61^{[9]}$ & $  5.10^{[14]}$ & $  2.07^{[13]}$ &&&&&&&&&\\
Ipswich Town &&& $  8.92^{[18]}$ & $  5.07^{[5]}$ &&&&&& $  1.73^{[22]}$ & $  6.55^{[19]}$ & $ 1.00^{[16]}$ \\
Leeds United & $   6.69^{[19]}$ & $ 41.21^{[15]}$ & $162.14^{[5]}$ & $ 48.46^{[4]}$ & $ 128.36^{[3]}$ & $  7.40^{[4]}$ & $ 19.34^{[5]}$ & $  2.23^{[11]}$ & $157.99^{[13]}$ & $ 55.27^{[5]}$ & $ 28.36^{[5]}$ & $36.52^{[17]}$ \\
Leicester City & $   2.28^{[18]}$ && $  1.00^{[20]}$ & $  9.38^{[13]}$ & $  16.05^{[8]}$ & $  3.01^{[10]}$ & $ 20.86^{[10]}$ & $  1.00^{[9]}$ && $  1.00^{[21]}$ &&\\
Liverpool & $  83.72^{[4]}$ & $273.92^{[5]}$ & $157.13^{[2]}$ & $ 38.34^{[3]}$ & $  35.24^{[4]}$ & $ 32.11^{[7]}$ & $195.77^{[3]}$ & $133.63^{[4]}$ & $ 88.34^{[3]}$ & $ 27.54^{[4]}$ & $ 32.56^{[8]}$ & $12.59^{[6]}$ \\
Manchester City & $  11.46^{[16]}$ & $  3.98^{[9]}$ && $  1.00^{[18]}$ &&&&& $  2.59^{[18]}$ & $ 10.35^{[17]}$ & $ 26.55^{[16]}$ & $14.47^{[9]}$ \\
Manchester United & $1261.10^{[3]}$ & $368.21^{[1]}$ & $206.47^{[3]}$ & $642.74^{[1]}$ & $1507.90^{[1]}$ & $119.75^{[1]}$ & $862.18^{[2]}$ & $305.69^{[1]}$ & $335.75^{[1]}$ & $320.33^{[2]}$ & $639.19^{[1]}$ & $34.48^{[1]}$ \\
Middlesbrough & $   5.60^{[11]}$ & $  7.94^{[11]}$ & $  3.12^{[12]}$ & $  7.07^{[14]}$ & $   6.70^{[12]}$ & $  2.06^{[9]}$ && $  2.02^{[19]}$ & $  5.38^{[12]}$ &&& $ 1.22^{[21]}$ \\
Newcastle United & $  32.07^{[5]}$ & $ 24.08^{[3]}$ & $  8.45^{[4]}$ & $ 12.33^{[11]}$ & $   4.94^{[11]}$ & $  7.53^{[13]}$ & $349.18^{[13]}$ & $ 87.90^{[2]}$ & $ 88.20^{[2]}$ & $139.16^{[6]}$ & $  7.55^{[3]}$ &\\
Norwich City &&&&&&&&&& $  9.31^{[20]}$ & $ 50.39^{[12]}$ & $ 6.15^{[3]}$ \\
Nottingham Forest &&&&&& $  1.93^{[20]}$ && $  7.69^{[20]}$ & $ 87.82^{[9]}$ & $ 11.28^{[3]}$ && $ 3.73^{[22]}$ \\
Oldham Athletic &&&&&&&&&&& $  4.95^{[21]}$ & $ 3.07^{[19]}$ \\
Portsmouth & $   2.49^{[13]}$ &&&&&&&&&&&\\
Queens Park Rangers &&&&&&&&& $ 17.08^{[19]}$ & $  8.83^{[8]}$ & $ 28.32^{[9]}$ & $16.35^{[5]}$ \\
Sheffield United &&&&&&&&&&& $ 14.60^{[20]}$ & $ 5.17^{[14]}$ \\
Sheffield Wednesday &&&&& $   2.09^{[19]}$ & $  2.55^{[12]}$ & $ 10.85^{[16]}$ & $  3.25^{[7]}$ & $ 11.87^{[15]}$ & $ 23.59^{[13]}$ & $ 17.28^{[7]}$ & $12.52^{[7]}$ \\
Southampton & $  11.97^{[12]}$ & $  4.53^{[8]}$ & $  3.47^{[11]}$ & $  2.43^{[10]}$ & $   6.37^{[15]}$ & $  1.00^{[17]}$ & $  5.77^{[12]}$ & $  2.26^{[16]}$ & $  9.06^{[17]}$ & $  5.06^{[10]}$ & $  4.02^{[18]}$ & $ 4.53^{[18]}$ \\
Sunderland && $  2.40^{[20]}$ & $ 11.81^{[17]}$ & $ 10.06^{[7]}$ & $  45.99^{[7]}$ &&& $  1.59^{[18]}$ &&&&\\
Swindon Town &&&&&&&&&&& $  1.00^{[22]}$ &\\
Tottenham Hotspur & $   2.73^{[14]}$ & $  6.21^{[10]}$ & $  4.24^{[9]}$ & $  5.73^{[12]}$ & $  12.71^{[10]}$ & $  1.79^{[11]}$ & $ 18.08^{[14]}$ & $ 14.83^{[10]}$ & $ 29.58^{[8]}$ & $  3.55^{[7]}$ & $ 29.34^{[15]}$ & $ 3.16^{[8]}$ \\
Watford &&&&& $   1.00^{[20]}$ &&&&&&&\\
West Bromwich Albion && $  1.27^{[19]}$ &&&&&&&&&&\\
West Ham United && $  4.36^{[18]}$ & $  2.51^{[7]}$ & $  6.45^{[15]}$ & $  19.28^{[9]}$ & $  5.16^{[5]}$ & $ 19.23^{[8]}$ & $  4.07^{[14]}$ & $  8.53^{[10]}$ & $  4.24^{[14]}$ & $  4.57^{[13]}$ &\\
Wimbledon &&&&& $   2.27^{[18]}$ & $  3.44^{[16]}$ & $ 20.56^{[15]}$ & $  9.50^{[8]}$ & $ 12.53^{[14]}$ & $ 23.28^{[9]}$ & $ 39.03^{[6]}$ & $ 6.92^{[12]}$ \\
Wolverhampton Wanderers & $   1.00^{[20]}$ &&&&&&&&&&&\\

\hline
$p$-value &
$0.25083$ & $0.02022$ & $0.16761$ & $0.92963$ & $0.99817$ & $0.76497$ & $0.01572$ & $0.03854$ & $0.0628$ & $0.12783$ & $0.45384$ & $0.00309$ \\
\hline
\end{tabular}
\end{sidewaystable}

\begin{sidewaystable}
    \centering
    \caption{\emph{La Liga}: 
    relative probabilities of the expectation of winning and final ranks (seasons from 2004-05
    to 2016-17). At bottom: $p$-value of the test \eqref{eq:LRtest}.}\label{TabP:ES2}
\begin{tabular}{|c*{13}{|@{\;}r@{\;}}|}
\hline
Name & 2016-17 & 2015-16 & 2014-15 & 2013-14 & 2012-13 & 2011-12 & 2010-11 & 2009-10 & 2008-09 & 2007-08 & 2006-07 & 2005-06 & 2004-05 \\
\hline
Alavés & $   2.17^{[9]}$ &&&&&&&&&&& $  1.00^{[18]}$ &\\
Albacete &&&&&&&&&&&&& $  2.55^{[20]}$ \\
Almería &&& $   6.08^{[19]}$ & $   1.51^{[17]}$ &&& $   2.42^{[20]}$ & $   3.41^{[13]}$ & $ 14.23^{[11]}$ & $  1.63^{[8]}$ &&&\\
Athletic Bilbao & $  80.11^{[7]}$ & $  17.39^{[5]}$ & $  48.23^{[7]}$ & $   4.57^{[4]}$ & $   5.64^{[12]}$ & $   12.89^{[10]}$ & $   5.11^{[6]}$ & $   4.16^{[8]}$ & $ 10.72^{[13]}$ & $  3.00^{[11]}$ & $  4.72^{[17]}$ & $  6.91^{[12]}$ & $ 15.58^{[9]}$ \\
Atlético Madrid & $2378.09^{[3]}$ & $ 365.54^{[3]}$ & $1410.15^{[3]}$ & $ 209.95^{[1]}$ & $  99.25^{[3]}$ & $   36.11^{[5]}$ & $  14.24^{[7]}$ & $  25.44^{[9]}$ & $ 96.75^{[4]}$ & $ 24.96^{[4]}$ & $ 19.50^{[7]}$ & $  5.33^{[10]}$ & $ 11.71^{[11]}$ \\
Barcelona & $5620.24^{[2]}$ & $4580.26^{[1]}$ & $3869.10^{[1]}$ & $5419.42^{[2]}$ & $5380.00^{[1]}$ & $11483.46^{[2]}$ & $1928.98^{[1]}$ & $1315.49^{[1]}$ & $237.27^{[1]}$ & $431.44^{[3]}$ & $787.60^{[2]}$ & $195.09^{[1]}$ & $185.73^{[1]}$ \\
Betis &&&&&& $    9.97^{[13]}$ &&& $  9.36^{[18]}$ & $  3.17^{[13]}$ & $  5.11^{[16]}$ & $ 21.37^{[14]}$ & $ 10.30^{[4]}$ \\
Celta de Vigo & $  21.74^{[13]}$ & $  32.01^{[6]}$ & $  26.99^{[8]}$ & $   1.48^{[9]}$ & $   2.64^{[17]}$ &&&&&& $ 13.58^{[18]}$ & $  5.34^{[6]}$ &\\
Cádiz &&&&&&&&&&&& $  1.82^{[19]}$ &\\
Córdoba &&& $   1.00^{[20]}$ &&&&&&&&&&\\
Deportivo La Coruña & $   3.74^{[16]}$ & $   3.35^{[15]}$ & $   1.90^{[16]}$ && $   5.41^{[19]}$ && $   1.65^{[18]}$ & $   9.08^{[10]}$ & $ 17.42^{[7]}$ & $  3.11^{[9]}$ & $ 10.04^{[13]}$ & $ 10.69^{[8]}$ & $ 32.49^{[8]}$ \\
Eibar & $   4.01^{[10]}$ & $   1.49^{[14]}$ & $   2.64^{[18]}$ &&&&&&&&&&\\
Elche &&& $   3.29^{[13]}$ & $   1.00^{[16]}$ &&&&&&&&&\\
Espanyol & $   4.23^{[8]}$ & $   7.03^{[13]}$ & $   6.68^{[10]}$ & $   5.05^{[14]}$ & $   1.47^{[13]}$ & $    5.08^{[14]}$ & $   4.41^{[8]}$ & $   7.22^{[11]}$ & $  7.90^{[10]}$ & $ 13.54^{[12]}$ & $  2.91^{[11]}$ & $ 11.15^{[15]}$ & $ 10.04^{[5]}$ \\
Getafe && $   1.43^{[19]}$ & $   6.20^{[15]}$ & $   3.96^{[13]}$ & $   3.17^{[10]}$ & $    4.39^{[11]}$ & $   8.58^{[16]}$ & $   5.08^{[6]}$ & $ 21.72^{[17]}$ & $  6.08^{[14]}$ & $ 12.29^{[9]}$ & $  6.17^{[9]}$ & $  1.00^{[13]}$ \\
Gimnàstic &&&&&&&&&&& $  1.00^{[20]}$ &&\\
Granada & $   2.25^{[20]}$ & $   1.00^{[16]}$ & $   4.89^{[17]}$ & $   2.76^{[15]}$ & $   1.00^{[15]}$ & $    1.00^{[17]}$ &&&&&&&\\
Hércules &&&&&&& $   2.77^{[19]}$ &&&&&&\\
Las Palmas & $   9.28^{[14]}$ & $   1.87^{[11]}$ &&&&&&&&&&&\\
Leganés & $   2.34^{[17]}$ &&&&&&&&&&&&\\
Levante && $   1.73^{[20]}$ & $   8.58^{[14]}$ & $   4.22^{[10]}$ & $   4.23^{[11]}$ & $   10.42^{[6]}$ & $   1.00^{[14]}$ &&& $  1.00^{[20]}$ & $  1.91^{[15]}$ && $  2.76^{[18]}$ \\
Mallorca &&&&& $  12.24^{[18]}$ & $    5.26^{[8]}$ & $  14.91^{[17]}$ & $  22.18^{[5]}$ & $ 56.13^{[9]}$ & $  6.19^{[7]}$ & $  4.67^{[12]}$ & $  2.59^{[13]}$ & $  5.56^{[17]}$ \\
Murcia &&&&&&&&&& $  1.00^{[19]}$ &&&\\
Málaga & $   8.91^{[11]}$ & $   4.31^{[8]}$ & $  17.20^{[9]}$ & $  19.44^{[11]}$ & $  22.60^{[6]}$ & $   16.59^{[4]}$ & $   3.33^{[11]}$ & $   2.96^{[17]}$ & $  1.00^{[8]}$ &&& $  8.41^{[20]}$ & $  8.08^{[10]}$ \\
Numancia &&&&&&&&& $  1.74^{[19]}$ &&&& $  1.45^{[19]}$ \\
Osasuna & $   1.00^{[19]}$ &&& $   1.41^{[18]}$ & $   4.61^{[16]}$ & $   12.17^{[7]}$ & $   2.63^{[9]}$ & $   5.11^{[12]}$ & $  9.35^{[15]}$ & $ 10.22^{[17]}$ & $ 10.57^{[14]}$ & $  2.70^{[4]}$ & $  9.25^{[15]}$ \\
Racing Santander &&&&&& $    5.42^{[20]}$ & $   2.29^{[12]}$ & $   5.50^{[16]}$ & $ 15.86^{[12]}$ & $  4.37^{[6]}$ & $  2.00^{[10]}$ & $  2.99^{[17]}$ & $  2.41^{[16]}$ \\
Rayo Vallecano && $   3.21^{[18]}$ & $   8.61^{[11]}$ & $   1.63^{[12]}$ & $   1.35^{[8]}$ & $    2.60^{[15]}$ &&&&&&&\\
Real Betis & $   6.69^{[15]}$ & $   3.93^{[10]}$ && $   9.01^{[20]}$ & $   4.13^{[7]}$ &&&&&&&&\\
Real Madrid & $6628.36^{[1]}$ & $3882.49^{[2]}$ & $7037.55^{[2]}$ & $2306.07^{[3]}$ & $3352.33^{[2]}$ & $ 3274.69^{[1]}$ & $ 652.95^{[2]}$ & $ 186.16^{[2]}$ & $595.32^{[2]}$ & $251.50^{[1]}$ & $133.10^{[1]}$ & $107.57^{[2]}$ & $ 43.03^{[2]}$ \\
Real Sociedad & $  16.07^{[6]}$ & $   7.46^{[9]}$ & $  28.13^{[12]}$ & $  24.77^{[7]}$ & $   3.72^{[4]}$ & $    3.79^{[12]}$ & $   1.05^{[15]}$ &&&& $  2.26^{[19]}$ & $  3.80^{[16]}$ & $  8.05^{[14]}$ \\
Recreativo &&&&&&&&& $  6.73^{[20]}$ & $  8.95^{[16]}$ & $  3.90^{[8]}$ &&\\
Sevilla & $  67.23^{[4]}$ & $  81.79^{[7]}$ & $ 143.24^{[5]}$ & $   9.33^{[5]}$ & $  14.15^{[9]}$ & $   34.39^{[9]}$ & $  16.85^{[5]}$ & $ 111.10^{[4]}$ & $244.40^{[3]}$ & $ 61.80^{[5]}$ & $119.42^{[3]}$ & $ 11.26^{[5]}$ & $ 28.54^{[6]}$ \\
Sporting Gijón & $   4.80^{[18]}$ & $   3.10^{[17]}$ &&&& $    4.14^{[19]}$ & $   1.01^{[10]}$ & $   1.04^{[15]}$ & $  1.08^{[14]}$ &&&&\\
Tenerife &&&&&&&& $   1.77^{[19]}$ &&&&&\\
Valencia & $   9.95^{[12]}$ & $  92.03^{[12]}$ & $  95.75^{[4]}$ & $  29.52^{[8]}$ & $  25.07^{[5]}$ & $  113.57^{[3]}$ & $  61.40^{[3]}$ & $  16.98^{[3]}$ & $ 33.03^{[6]}$ & $ 59.19^{[10]}$ & $127.03^{[4]}$ & $ 34.23^{[3]}$ & $333.74^{[7]}$ \\
Valladolid &&&& $   1.97^{[19]}$ & $   4.61^{[14]}$ &&& $   1.88^{[18]}$ & $  7.36^{[16]}$ & $  1.76^{[15]}$ &&&\\
Villarreal & $  44.54^{[5]}$ & $  29.28^{[4]}$ & $  39.35^{[6]}$ & $   8.48^{[6]}$ && $   26.77^{[18]}$ & $  13.09^{[4]}$ & $  19.90^{[7]}$ & $325.93^{[5]}$ & $ 62.77^{[2]}$ & $ 27.58^{[5]}$ & $ 42.58^{[7]}$ & $ 11.24^{[3]}$ \\
Xerez &&&&&&&& $   1.00^{[20]}$ &&&&&\\
Zaragoza &&&&& $   2.66^{[20]}$ & $    6.82^{[16]}$ & $   2.30^{[13]}$ & $   6.17^{[14]}$ && $ 13.83^{[18]}$ & $  5.98^{[6]}$ & $  6.81^{[11]}$ & $  6.91^{[12]}$ \\

\hline
$p$-value &
$0.66473$ & $0.0913$ & $0.93225$ & $0.00616$ & $0.06318$ & $0.03868$ & $0.06037$ & $0.94837$ & $0.09938$ & $0.03276$ & $0.37931$ & $0.10886$ & $0.15635$  \\
\hline
\end{tabular}
\end{sidewaystable}

\begin{sidewaystable}
    \centering
    \caption{\emph{La Liga}: 
    relative probabilities of the expectation of winning and final ranks (seasons from 1992-93
    to 2003-04). At bottom: $p$-value of the test \eqref{eq:LRtest}.}\label{TabP:ES1}
\begin{tabular}{|c*{12}{|@{\;}r@{\;}}|}
\hline
Name & 2003-04 & 2002-03 & 2001-02 & 2000-01 & 1999-00 & 1998-99 & 1997-98 & 1996-97 & 1995-96 & 1994-95 & 1993-94 & 1992-93 \\
\hline
Alavés && $  8.75^{[19]}$ & $ 18.91^{[7]}$ & $ 19.87^{[10]}$ & $  3.84^{[6]}$ & $  2.00^{[16]}$ &&&&&&\\
Albacete & $  1.00^{[14]}$ &&&&&&&& $ 14.59^{[20]}$ & $  4.29^{[17]}$ & $  4.37^{[13]}$ & $   3.30^{[17]}$ \\
Athletic Bilbao & $ 20.24^{[5]}$ & $  3.20^{[7]}$ & $  8.34^{[9]}$ & $ 19.72^{[12]}$ & $ 31.91^{[11]}$ & $ 30.87^{[8]}$ & $ 25.36^{[2]}$ & $  31.29^{[6]}$ & $ 69.50^{[15]}$ & $ 19.13^{[8]}$ & $ 14.20^{[5]}$ & $   6.14^{[8]}$ \\
Atlético Madrid & $  3.55^{[7]}$ & $  4.50^{[12]}$ &&& $ 12.42^{[19]}$ & $ 63.44^{[13]}$ & $ 37.45^{[7]}$ & $ 590.17^{[5]}$ & $ 86.98^{[1]}$ & $ 10.14^{[14]}$ & $ 55.44^{[12]}$ & $ 165.05^{[6]}$ \\
Barcelona & $ 89.35^{[2]}$ & $105.44^{[6]}$ & $ 79.95^{[4]}$ & $ 88.26^{[4]}$ & $365.60^{[2]}$ & $ 84.32^{[1]}$ & $485.47^{[1]}$ & $1218.84^{[2]}$ & $485.30^{[3]}$ & $903.52^{[4]}$ & $567.63^{[1]}$ & $1634.99^{[1]}$ \\
Betis & $ 19.69^{[9]}$ & $ 31.53^{[8]}$ & $  6.16^{[6]}$ && $  6.47^{[18]}$ & $ 11.27^{[11]}$ & $ 35.83^{[8]}$ & $ 118.98^{[4]}$ & $125.68^{[8]}$ & $  5.28^{[3]}$ &&\\
Celta de Vigo & $ 30.65^{[19]}$ & $ 54.20^{[4]}$ & $ 64.62^{[5]}$ & $ 56.76^{[6]}$ & $ 55.26^{[7]}$ & $ 14.13^{[5]}$ & $  8.63^{[6]}$ & $  33.16^{[16]}$ & $ 13.12^{[11]}$ & $  7.05^{[13]}$ & $  2.64^{[15]}$ & $   3.17^{[11]}$ \\
Compostela &&&&&&& $  6.58^{[17]}$ & $  18.71^{[11]}$ & $ 16.52^{[10]}$ & $  1.00^{[16]}$ &&\\
Cádiz &&&&&&&&&&&& $   2.93^{[19]}$ \\
Deportivo La Coruña & $120.79^{[3]}$ & $ 89.61^{[3]}$ & $ 88.91^{[2]}$ & $112.62^{[2]}$ & $ 36.69^{[1]}$ & $ 15.44^{[6]}$ & $ 36.64^{[12]}$ & $ 346.88^{[3]}$ & $448.59^{[9]}$ & $313.07^{[2]}$ & $ 64.34^{[2]}$ & $  10.14^{[3]}$ \\
Espanyol & $  5.13^{[16]}$ & $  5.69^{[17]}$ & $  7.98^{[14]}$ & $ 19.99^{[9]}$ & $ 49.00^{[14]}$ & $ 13.65^{[7]}$ & $  8.26^{[10]}$ & $ 208.69^{[12]}$ & $103.75^{[4]}$ & $  9.55^{[6]}$ && $   4.31^{[18]}$ \\
Extremadura &&&&&& $  1.90^{[17]}$ && $   1.00^{[19]}$ &&&&\\
Hércules &&&&&&&& $   3.56^{[21]}$ &&&&\\
Las Palmas &&& $  3.13^{[18]}$ & $  1.55^{[11]}$ &&&&&&&&\\
Lleida &&&&&&&&&&& $  1.99^{[19]}$ &\\
Logroñés &&&&&&&& $   7.09^{[22]}$ && $  5.52^{[20]}$ & $  3.21^{[16]}$ & $   4.86^{[15]}$ \\
Mallorca & $  6.64^{[11]}$ & $  5.77^{[9]}$ & $ 45.48^{[16]}$ & $ 19.49^{[3]}$ & $ 29.71^{[10]}$ & $ 18.82^{[3]}$ & $  3.04^{[5]}$ &&&&&\\
Murcia & $  2.17^{[20]}$ &&&&&&&&&&&\\
Málaga & $  8.43^{[10]}$ & $ 28.20^{[13]}$ & $ 10.45^{[10]}$ & $ 10.70^{[8]}$ & $  1.48^{[12]}$ &&&&&&&\\
Mérida &&&&&&& $  1.00^{[19]}$ && $  8.50^{[21]}$ &&&\\
Numancia &&&& $  3.10^{[20]}$ & $  1.00^{[17]}$ &&&&&&&\\
Osasuna & $ 10.89^{[13]}$ & $  2.91^{[11]}$ & $  2.67^{[17]}$ & $  1.00^{[15]}$ &&&&&&& $  3.41^{[20]}$ & $  10.98^{[10]}$ \\
Oviedo &&&& $  4.72^{[18]}$ & $  3.09^{[16]}$ & $  2.88^{[14]}$ & $  3.34^{[18]}$ & $  30.91^{[17]}$ & $ 47.18^{[14]}$ & $  8.33^{[9]}$ & $  3.93^{[9]}$ & $   9.27^{[16]}$ \\
Racing Santander & $  4.12^{[17]}$ & $  3.44^{[16]}$ && $  9.84^{[19]}$ & $  4.15^{[15]}$ & $  4.39^{[15]}$ & $  3.16^{[14]}$ & $  19.97^{[13]}$ & $ 13.31^{[17]}$ & $  4.87^{[12]}$ & $  1.00^{[8]}$ &\\
Rayo Vallecano && $ 11.52^{[20]}$ & $  2.48^{[11]}$ & $  8.38^{[14]}$ & $  1.77^{[9]}$ &&& $  11.22^{[18]}$ & $  7.45^{[19]}$ && $  1.82^{[17]}$ & $   1.00^{[14]}$ \\
Real Burgos &&&&&&&&&&&& $   6.10^{[20]}$ \\
Real Madrid & $303.72^{[4]}$ & $314.39^{[1]}$ & $137.65^{[3]}$ & $205.37^{[1]}$ & $105.98^{[5]}$ & $212.02^{[2]}$ & $260.96^{[4]}$ & $ 551.17^{[1]}$ & $414.27^{[6]}$ & $210.53^{[1]}$ & $140.98^{[4]}$ & $ 395.07^{[2]}$ \\
Real Sociedad & $ 87.71^{[15]}$ & $ 22.52^{[2]}$ & $  4.91^{[13]}$ & $ 15.52^{[13]}$ & $ 23.37^{[13]}$ & $ 36.02^{[10]}$ & $ 13.89^{[3]}$ & $ 132.98^{[8]}$ & $ 30.54^{[7]}$ & $  5.44^{[11]}$ & $  8.34^{[11]}$ & $  12.75^{[13]}$ \\
Recreativo && $  1.00^{[18]}$ &&&&&&&&&&\\
Salamanca &&&&&& $  5.62^{[20]}$ & $  2.61^{[15]}$ && $  1.00^{[22]}$ &&&\\
Sevilla & $ 14.56^{[6]}$ & $ 10.94^{[10]}$ & $  2.54^{[8]}$ && $  1.12^{[20]}$ &&& $  31.51^{[20]}$ & $ 91.02^{[12]}$ & $ 32.94^{[5]}$ & $ 12.77^{[6]}$ & $   9.33^{[7]}$ \\
Sporting Gijón &&&&&&& $  1.56^{[20]}$ & $  15.51^{[15]}$ & $  5.34^{[18]}$ & $  3.38^{[18]}$ & $  3.58^{[14]}$ & $  18.63^{[12]}$ \\
Tenerife &&& $  1.00^{[19]}$ &&& $  7.85^{[19]}$ & $ 13.13^{[16]}$ & $ 193.01^{[9]}$ & $ 19.54^{[5]}$ & $ 13.85^{[15]}$ & $ 17.78^{[10]}$ & $   8.56^{[5]}$ \\
Valencia & $101.31^{[1]}$ & $345.94^{[5]}$ & $108.68^{[1]}$ & $324.54^{[5]}$ & $ 29.22^{[3]}$ & $ 22.12^{[4]}$ & $  8.81^{[9]}$ & $ 589.22^{[10]}$ & $ 61.41^{[2]}$ & $ 48.70^{[10]}$ & $116.95^{[7]}$ & $  66.93^{[4]}$ \\
Valladolid & $  8.23^{[18]}$ & $  8.49^{[14]}$ & $  5.91^{[12]}$ & $ 17.97^{[16]}$ & $  9.43^{[8]}$ & $ 12.38^{[12]}$ & $  6.60^{[11]}$ & $  33.91^{[7]}$ & $  2.57^{[16]}$ & $  2.52^{[19]}$ & $  3.00^{[18]}$ &\\
Villarreal & $  9.35^{[8]}$ & $  4.45^{[15]}$ & $ 14.77^{[15]}$ & $  1.87^{[7]}$ && $  1.00^{[18]}$ &&&&&&\\
Zaragoza & $  1.39^{[12]}$ && $  8.48^{[20]}$ & $ 46.12^{[17]}$ & $ 25.12^{[4]}$ & $ 11.76^{[9]}$ & $  5.46^{[13]}$ & $  36.40^{[14]}$ & $ 78.37^{[13]}$ & $ 51.10^{[7]}$ & $  3.22^{[3]}$ & $  17.42^{[9]}$ \\

\hline
$p$-value &
$0.00031$ & $0.0308$ & $0.02556$ & $0.00962$ & $0.00973$ & $0.114$ & $0.2829$ & $0.1186$ & $0.10968$ & $0.1226$ & $0.04398$ & $0.2296$ 
\\
\hline
\end{tabular}
\end{sidewaystable}


\begin{thebibliography}{10}

\bibitem{A1art}
G.~Aletti.
\newblock Generation of discrete random variables in scalable frameworks.
\newblock {\em ArXiv e-prints}, Nov. 2016.

\bibitem{Cor2014}
L.~Corain, V.~B. Melas, A.~Pepelyshev, and L.~Salmaso.
\newblock New insights on permutation approach for hypothesis testing on
  functional data.
\newblock {\em Advances in Data Analysis and Classification}, 8(3):339--356,
  Sep 2014.

\bibitem{Hub09}
P.~J. Huber and E.~M. Ronchetti.
\newblock {\em Robust Statistics}.
\newblock Wiley Series in Probability and Statistics. John Wiley \& Sons, Inc.,
  Hoboken, NJ, second edition, 2009.

\bibitem{Las13}
J.~Lasek, Z.~Szl{\'a}vik, and S.~Bhulai.
\newblock The predictive power of ranking systems in association football.
\newblock {\em International Journal of Applied Pattern Recognition},
  1(1):27--46, 2013.

\bibitem{Pes12}
F.~Pesarin and L.~Salmaso.
\newblock A review and some new results on permutation testing for multivariate
  problems.
\newblock {\em Statistics and Computing}, 22(2):639--646, Mar 2012.

\bibitem{Pie12}
J.~Pierce.
\newblock {\em An Introduction to Information Theory: Symbols, Signals and
  Noise}.
\newblock Dover Books on Mathematics. Dover Publications, 2012.

\bibitem{Sab15}
J.~A. Sabourin, W.~Valdar, and A.~B. Nobel.
\newblock A permutation approach for selecting the penalty parameter in
  penalized model selection.
\newblock {\em Biometrics}, 71(4):1185--1194, 2015.

\bibitem{Stern90}
H.~Stern.
\newblock Models for distributions on permutations.
\newblock {\em Journal of the American Statistical Association},
  85(410):558--564, 1990.

\bibitem{Kenneth09}
K.~E. Train.
\newblock {\em Discrete Choice Methods with Simulation}.
\newblock Cambridge University Press, Cambridge, second edition, 2009.

\bibitem{Ush17}
N.~G. Ushakov and V.~G. Ushakov.
\newblock Permutation tests for homogeneity based on some characterizations.
\newblock {\em Communications in Statistics - Theory and Methods},
  46(15):7692--7702, 2017.

\end{thebibliography}
\end{document}